\documentclass[letterpaper, 10 pt, conference]{ieeeconf}
\usepackage{amsmath,amsfonts,amssymb,mathrsfs}
\usepackage{graphicx}
\usepackage{cite}

\usepackage{float}
\usepackage{color}
\usepackage{subcaption}
\usepackage{eucal} 
\usepackage{url}

\usepackage[ruled]{algorithm2e}
\usepackage{footnote}
\makesavenoteenv{algorithm}
\usepackage{algpseudocode}
\usepackage{soul}
\usepackage{multirow}
\usepackage{multicol}
\usepackage[dvipsnames]{xcolor}

\newtheorem{assume}{Assumption}[section]
\newtheorem{theorem}{Theorem}[section]
\newtheorem{remark}{Remark}[section]
\newtheorem{lemma}{Lemma}[section]

\IEEEoverridecommandlockouts
\title{\LARGE \bf
	Data-Driven Cooperative Adaptive Cruise Control for Unknown Nonlinear Vehicle Platoons
}

\author{Jianglin Lan
	\thanks{This work was supported by a Leverhulme Trust Early Career 
		Fellowship under Award ECF-2021-517.}
	\thanks{The author is with the School of Engineering, University 
		of Glasgow, Glasgow G12 8QQ, UK {\tt\small jianglin.lan@glasgow.ac.uk}}
}

\begin{document}
	\maketitle

\begin{abstract}
This paper studies cooperative adaptive cruise control (CACC) for vehicle platoons with consideration of the unknown nonlinear vehicle dynamics that are normally ignored in the literature. A unified data-driven CACC design is proposed for platoons of pure automated vehicles (AVs) or of mixed AVs and human-driven vehicles (HVs). The CACC leverages online-collected sufficient data samples of vehicle accelerations, spacing and relative velocities. The data-driven control design is formulated as a semidefinite program (SDP) that can be solved efficiently using off-the-shelf solvers. The efficacy and advantage of the proposed CACC are demonstrated through a comparison with the classic adaptive cruise control (ACC) method on a platoon of pure AVs and a mixed platoon under a representative aggressive driving profile. 
\end{abstract}

\section{Introduction}
	Vehicle platoon refers to a convoy of vehicles that travel at the same (longitudinal) velocity  
	whilst maintaining a desired inter-vehicular distance. Previous studies  \cite{Arem+06,Ploeg+11,Yang+22} revealed 
	that vehicle platoons have a great potential in reducing traffic congestion, 
	accidents and fuel consumption. 
	An effective control strategy is the key to establish vehicle 
	platoons and has attracted much research attention. 
	
	The classic adaptive cruise control (ACC) \cite{XiaoGao10}, which has been 
	well-developed and available on the market, enables the ego vehicle to follow 
	its preceding vehicle. But ACC is not sufficient to establish a stable vehicle 
	platoon \cite{Shladover+12,Gunter+21}. This motivated the development 
	of a more advanced platooning control strategy, the cooperative adaptive cruise control 
	(CACC), by using 
	vehicle information shared in the platoon through vehicle-to-vehicle (V2V) wireless 
	communications. Many works have shown that CACC can effectively establish stable platoons 
	with pure automated vehicles (AVs) \cite{guanetti2018control,lan2020min}.
	Several works \cite{chen2020mixed,feng2019robust,hajdu2019robust,Lan+21d} have shown that CACC are still effective for a more challenging case when a platoon has mixed AVs and human-driven vehicles (HVs), where the car-following behaviours of HVs are known to be different. 
	
	However, most existing CACC approaches hinge on linear vehicle  
	models (\textit{i.e.} the point-mass model) without considering the 
	nonlinear dynamics. 
	The nonlinear dynamics could be neglected for small-size passenger cars at low speeds but not for medium-size or large-size vehicles such as trucks and heavy-duty vehicles \cite{Shen+22}.
	This limits the practical applicability of existing CACC designs in 
	real traffic systems. Some works have studied 
	CACC for nonlinear vehicle platoons through feedback linearisation 
	based on accurately known vehicle parameters 
	\cite{Ghasemi+13,Hu+20,Shen+22}. However, some vehicle parameters may be changing or	unknown for CACC design. Another realistic concern is that in the heterogeneous platoon 
	the parameters of other vehicles may be unknown to the ego AV.
	These make it necessary to develop CACC for platoons with unknown nonlinear 
	vehicle dynamics.
	
	Data-driven CACC that do not rely on the known vehicle dynamics have 
	been recently developed for mixed vehicle platoons, using the methods like adaptive 
	dynamic programming (ADP) \cite{gao2016data,huang2019connected,huang2020learning}, data-driven model predictive control (MPC) \cite{Lan+21d,Wang+22}, or reinforcement learning \cite{chu2019model}. These methods are only applied to mixed platoons with linearised vehicle models.
	The CACC designs for pure AVs with nonlinear vehicle models were achieved by combining model-based control with linearisation of the nonlinear dynamics based on a parameter estimation \cite{Zhu+20} or data-driven feedforward control \cite{Gao+17}.

	This paper proposes a novel data-driven CACC to address unknown nonlinear vehicle dynamics. The design takes inspiration from \cite{Strasser+21,Persis+23} to represent the unknown platoon dynamics as a polynomial system. The main contributions are summarised as follows: 
	\begin{enumerate}
		\item A strategy is developed for learning a nonlinear CACC controller from vehicle data samples through solving a semidefinite program (SDP).
		\item The obtained CACC ensures the $H_\infty$ robust performance of the platoon against disturbances such as air/mechanical drags and reference velocity changes. In the data-driven control literature, $H_\infty$ robustness is either not addressed \cite{Persis+23} or studied using sum-of-squares optimisation \cite{Strasser+21}.
		\item The proposed strategy is applicable to both pure AVs platoons and mixed platoons.
	\end{enumerate}

	The rest of this paper is structured as follows.
	Section \ref{sec:problem statement} describes the platoon model, Section \ref{sec:control for AV platoon} 
	presents the data-driven CACC design for pure AVs and Section~\ref{subsec:mixed platoon} applies it to mixed platoons, Section \ref{sec:simulation} 
	reports the simulation results, and Section \ref{sec:conclusion} draws the 
	conclusions.
	
\textit{Notations}: 
The interval $[a,b]$ is the set of integers from $a$ to $b$. $\mathrm{diag}(V_1, \cdots, V_n)$ denotes a block diagonal matrix with the main diagonals $V_1, \cdots, V_2$. 
$I_n$ is a $n \times n$ identity matrix, $\mathbf{1}_n$ is a column of $n$  ones, and $\mathbf{0}$ is a zero matrix whose dimension is omitted unless necessary to be given.

	\section{Platoon Modelling and Preliminaries} \label{sec:problem statement}
	This paper considers a platoon with $n$ AVs 
	equipped with V2V wireless communication devices. As in the 
	literature, the design focuses on
	controlling the longitudinal dynamics of AVs to establish a platoon. To this end, a platooning error system needs to be built.
	The longitudinal dynamics of the $i$-th AV, $i \in [1,n]$, can be characterised by \cite{Ghasemi+13}:
	\begin{eqnarray}\label{eq:nonlinear vehicle dynamics}
		\begin{split}
			\dot{p}_i &= v_i, \\
			\dot{v}_i &= a_i, \\
			\dot{a}_i &= f_i(v_i,a_i) + u_i/(\tau_i m_i),
		\end{split}
	\end{eqnarray}
	where $p_i$, $v_i$, and $u_i$ are the vehicle position, velocity, and control effort, respectively.
	$f_i(v_i,a_i) = - ( a_i + R_i v_i^2 + d_i/m_i )/\tau_i - 2 R_i v_i a_i.
	$
	$R_i = \sigma_i M_i c_i/(2 m_i)$ is the air resistance.
	$\tau_i$, $\sigma_i$, $M_i$, $c_i$, $d_i$, and $m_i$ are the 
	engine time constant, specific mass of the air, cross-sectional area, drag 
	coefficient, mechanical drag, and mass of the vehicle, respectively. 
	
	To facilitate the platooning control design, the nonlinear 
	model~\eqref{eq:nonlinear vehicle dynamics} is commonly linearised in the 
	literature (see \textit{e.g.,} \cite{Ghasemi+13,Hu+20,Shen+22}) using the feedback 
	linearisation control law
	\begin{equation}\label{eq:feedback linearisation}
		u_i = m_i \bar{u}_i + 0.5 \sigma M_i c_i v_i^2 + d_i + \tau_i \sigma M_i c_i 
		v_i a_i,
	\end{equation} 
	where $\bar{u}_i$ is the new control signal to be designed. However, \eqref{eq:feedback linearisation} is 
	applicable only when all the vehicle parameters $\tau_i$, $\sigma_i$, $M_i$, 
	$c_i$, $d_i$ and $m_i$ are known. This is practically restrictive or 
	even unrealistic because the parameters change with the driving environments 
	such as payload, road conditions, weather, etc. It thus motivates us to develop a data-driven CACC for \eqref{eq:nonlinear vehicle dynamics} with unknown
	$\tau_i$, $\sigma_i$, $M_i$, $c_i$, $d_i$ and $m_i$.
	
	This paper aims to design $u_i$ for each AV to drive at a given desired 
	constant velocity $v^* > 0$ whilst maintaining a desired inter-vehicular
	distance $h^* > 0$. For the leader, assuming that there is a virtual vehicle 
	ahead of it driving at $v^*$, then the inter-vehicular distance 
	is defined as $h_1 = v^* t - p_1$. 
	For $i \in [2, n]$, the inter-vehicular distance is defined as $h_i = p_{i-1} - p_i$.  
	Define the spacing error as $\tilde{h}_i = h_i - h^*$ and velocity error 
	as $\tilde{v}_i = v_i - v^*$, $i \in [1,n]$. Then the platooning error system of the 
	$i$-th vehicle is derived as
	\begin{equation}\label{eq:error system 1}
		\begin{split}
			\dot{\tilde{h}}_i &= \tilde{v}_{i-1} - \tilde{v}_i, \\
			\dot{\tilde{v}}_i &= a_i, \\
			\dot{a}_i &= g_i(\tilde{v}_i, a_i) + u_i/(\tau_i m_i)  + w_i,
		\end{split}	 	
	\end{equation} 
	where 
	$g_i(\tilde{v}_i, a_i) 
	= -2R_i v^* \tilde{v}_i/\tau_i  - (1 + 2\tau_i R_i 
	v^*)a_i /\tau_i  - R_i \tilde{v}_i^2/\tau_i  - 2 R_i \tilde{v}_i a_i$,
	$w_i = -(m_i R_i (v^*)^2  + d_i)/(\tau_i m_i)$ and $\tilde{v}_0 = 0$.
	
	The platooning error system \eqref{eq:error system 1} can be rewritten as
	\begin{align}\label{eq:error system 2}
		\underbrace{\begin{bmatrix}
				\dot{\tilde{h}}_i \\ \dot{\tilde{v}}_i \\ \dot{a}_i
		\end{bmatrix}}_{\dot{x}_i}
		&=
		\underbrace{\begin{bmatrix}
				0 & -1 & 0 \\ 0 & 0 & 1 \\ 0 & -2R_i v^*/\tau_i & - (1 + 2\tau_i R_i v^*)/\tau_i
		\end{bmatrix}}_{A_i}	
		\underbrace{\begin{bmatrix}
				\tilde{h}_i \\ \tilde{v}_i 	\\ a_i
		\end{bmatrix}}_{x_i} \nonumber\\
		& +
		\underbrace{\begin{bmatrix}
				0 & 1 & 0 \\0 & 0 & 0 \\0 & 0 & 0
		\end{bmatrix}}_{C_i}
		\underbrace{\begin{bmatrix}
				\tilde{h}_{i-1} \\ \tilde{v}_{i-1} \\ a_{i-1}
		\end{bmatrix}}_{x_{i-1}} +
		\underbrace{\begin{bmatrix}
				0 & 0 \\ 0 & 0 \\ - 2 R_i & - R_i/\tau_i
		\end{bmatrix}}_{E_i} 
		\underbrace{\begin{bmatrix}
				\tilde{v}_i a_i \\	\tilde{v}_i^2
		\end{bmatrix}}_{Q_i(x_i)}
		\nonumber\\
		&+
		\underbrace{\begin{bmatrix}
				0 \\ 0 \\ 1/(\tau_i m_i)
		\end{bmatrix}}_{B_i} u_i
		+
		\underbrace{\begin{bmatrix}
				0 \\ 0 \\ 1
		\end{bmatrix}}_{D_i} w_i.
	\end{align} 
	
	Define 
	$x \!=\! [x_1^\top, \cdots, x_n^\top]^\top$.
	The overall platooning error system is derived as a polynomial system
	\begin{equation} \label{eq:error system 4}
		\dot{x} = A_c Z(x) + B_c u + D_c w,
	\end{equation} 
	where $A_c = [A_{c1} ~ A_{c2}]$, $Z(x) = [x^\top,~Q(x)^\top]^\top$, $Q(x) = [Q_1(x_1)^\top, \cdots, Q_n(x_n)^\top]^\top$,
	$w = [w_1, \cdots, w_n]^\top$,  and
	\begin{align}
		A_{c1} &= \begin{bmatrix}
			A_1 	   & \mathbf{0} &  \cdots & 
			\mathbf{0}\\
			C_2 	   & A_2  		&  \cdots & 
			\mathbf{0}\\
			\vdots & \ddots & \ddots & \vdots\\
			\mathbf{0} & \mathbf{0}  & C_n       	& A_n  
		\end{bmatrix}, ~
		A_{c2} = \mathrm{diag}(E_1, \cdots, E_n),
		\nonumber\\
		B_c &= \mathrm{diag}(B_1, \cdots, B_n), ~
		D_c = \mathrm{diag}(D_1, \cdots, D_n).
		\nonumber
	\end{align} 
	The system dimensions are $x \in \mathbb{R}^{n_x \times 1}$, $u \in 
	\mathbb{R}^{n_u \times 1}$, $w \in \mathbb{R}^{n_w \times 1}$, $A_c \in 
	\mathbb{R}^{n_x \times n_z}$, 
	$Z(x) \in \mathbb{R}^{n_z \times 1}$, $B_c \in \mathbb{R}^{n_x \times n_u}$, and $D_c \in \mathbb{R}^{n_x \times n_w}$, 
	with $n_x = 3n$, $n_u = n_w = n$ and $n_z = 5n$. 
	
	Discretizing \eqref{eq:error system 4} using the forward Euler method with a sample time $t_s$ gives the discrete-time polynomial system 
	\begin{equation} \label{eq:sys for design}
		x(k+1) = A Z(x(k)) + B u(k) + D w(k),
	\end{equation}
	where $k$ is the sampling step, 
	$A = [I_{n_x}, \mathbf{0}_{n_x \times (n_z-n_x)}] + t_s A_c$, $B = t_s B_c$ and $D = t_s D_c$. 
	
	Since the matrices $A$ and $B$ in \eqref{eq:sys for design} are unknown, the existing model-based CACC designs \cite{Li+17d,guanetti2018control,lan2020min}  are inapplicable. 
	This paper will design a data-driven CACC controller $u(k) = K Z(x(k))$ with a constant gain $K 
	\in \mathbb{R}^{n_u \times n_z}$ to stabilise the platoon error system \eqref{eq:sys for design}. This ensures the platoon 
	travel at the desired velocity $v^*$ whilst keeping the desired 
	vehicular gap $h^*$ between any two consecutive AVs.
	
	The following assumptions are made for the disturbance $w$ and nonlinearity $Q(x)$, which are essential for designing a 
	data-driven controller to ensure robust stability of \eqref{eq:sys for design}.
	\begin{assume}\label{assume:disturbance}
		$|w| \leq \delta \times \mathbf{1}_{n_w}$ for some known $\delta > 0$.	
	\end{assume} 
	Assumption \ref{assume:disturbance} is reasonable because the parameters 
	$m_i$, $R_i$, $d_i$, $\tau_i$ and $v^*$ are all 
	physically bounded. 
	Suppose that $\underline{m}_i \leq m_i \leq \overline{m}_i$, $\underline{R}_i \leq R_i \leq \overline{R}_i$, $\underline{d}_i \leq d_i \leq \overline{d}_i$ and $\underline{\tau}_i \leq \tau_i \leq \overline{\tau}_i$, $i \in [1,n_w]$. The following relations hold:
	$|w_i| = (m_i R_i (v^*)^2  + d_i)/(\tau_i m_i) \leq (\overline{m}_i \overline{R}_i (v^*)^2 + \overline{d}_i)/(\underline{\tau}_i \underline{m}_i) =: \delta_i.$ 
	Hence, the value of $\delta$ is chosen as $\delta = \max_{i \in [1,n_w]} 
	\delta_i$. 
	
	\begin{assume}\label{assume:nonlinearity}
		$
		\underset{|x| \rightarrow 0}{\lim} |Q(x)|/|x| = 0.
		$
	\end{assume}
	Assumption \ref{assume:nonlinearity} 
	indicates that the nonlinear function $Q(x)$ approaches the origin faster than 
	the state $x$. This is true because $Q(x)$ contains multiplications of the elements in $x$, which is easily seen from the definition of $Q(x)$ in \eqref{eq:error system 4} and \eqref{eq:sys for design}. 
	Assumption \ref{assume:nonlinearity} ensures that the linear dynamics dominate the nonlinear 
	dynamics around the origin.

	\section{Data-driven CACC design}\label{sec:control for AV platoon}
	To design the data-driven CACC, we first derive a data-based representation of the platooning error 
	system \eqref{eq:sys for design} using online collected vehicle data.  
	At the start of forming the platoon (before the data-driven CACC controller is designed), each AV uses the classic ACC controller \cite{Shladover+12} to maintain vehicle safety. Then a total number of $T$ samples of vehicle data (including position, velocity, acceleration and control signal) are collected. 
	The collected samples satisfy
	\begin{equation}\label{eq:data dynamics}
		x(s+1) = A Z(x(s)) + B u(s) + D w(s), ~s \in [0,T-1].	
	\end{equation}
	These samples are grouped into the data sequences:
	\begin{subequations}\label{eq:data seq1}	
		\begin{align}
			\setlength{\parindent}{0in}		
			U_0 \!&=\! [u(0), u(1), \dots, u(T-1)] \in \mathbb{R}^{n_u \times T}, \\
			X_0 \!&=\! [x(0), x(1), \dots, x(T-1)] \in \mathbb{R}^{n_x \times T}, \\
			X_1 \!&=\! [x(1), x(2), \dots, x(T)] \in \mathbb{R}^{n_x \times T}, \\
			\!Z_0 \!&=\! [Z(x(0)), Z(x(1)), \dots, Z(x(T-1))] \!\in\! \mathbb{R}^{n_z \times T}. 
		\end{align}	
	\end{subequations}
	Furthermore, let the sequence of unknown disturbance be
	\begin{equation}\label{eq:data seq2}
		W_0 = [w(0),w(1),\dots,w(T-1)] \in \mathbb{R}^{n_w \times T}.
	\end{equation}
	
	By using \eqref{eq:data seq1} and \eqref{eq:data seq2}, we take inspiration from \cite[Lemma 2]{Persis+23} and derive
	a data-based representation of the platooning error system \eqref{eq:sys for design} in 
	Lemma~\ref{lemma:data-based closed-loop sys}.
	\begin{lemma}\label{lemma:data-based closed-loop sys}
		If there exist matrices $K \in \mathbb{R}^{n_u \times n_z}$ and $G \in \mathbb{R}^{T \times n_z}$	satisfying
		\begin{equation}\label{eq:lemma data closed 1}
			\begin{bmatrix}
				K \\ I_{n_z}
			\end{bmatrix}	
			=
			\begin{bmatrix}
				U_0 \\ Z_0
			\end{bmatrix}	
			G,
		\end{equation}
		then the platooning error system \eqref{eq:sys for design} under the controller $u(k) = K Z(x(k))$ 
		has the closed-loop dynamics
		\begin{equation}\label{eq:lemma data closed 2}
			x(k+1) = \bar{A} x(k) + \bar{E} Q(x(k)) + D w(k),
		\end{equation}
		where $\bar{A} = (X_1 - D W_0) G_1$, $\bar{E} = (X_1 - D W_0) G_2$, and $G = 
		[G_1 ~ 
		G_2]$ with $G_1 \in \mathbb{R}^{T \times n_x}$ and $G_2 \in \mathbb{R}^{T 
			\times (n_z-n_x)}$.
	\end{lemma}
	\begin{proof}
		Substituting $u(k) = K Z(x(k))$ into \eqref{eq:sys for design} and using 
		\eqref{eq:lemma data closed 1} results in
		\begin{align}\label{eq:lemma data closed pf1}
			x(k+1) &= [B ~ A] \begin{bmatrix} K \\ I_{n_z} \end{bmatrix} Z(x(k)) + D w(k) \nonumber\\
			&= [B ~ A] \begin{bmatrix} U_0 \\ Z_0 \end{bmatrix} G Z(x(k)) + D w(k) \nonumber\\
			&= (A Z_0 + B U_0) G Z(x(k)) + D w(k).
		\end{align}
		Since the data sequences $U_0$, $X_0$, $X_1$, $Z_0$ and $D_0$ satisfy \eqref{eq:data dynamics}, the relation $X_1 = A Z_0 + B U_0 + D W_0$ holds. Applying this to \eqref{eq:lemma data closed pf1} and 
		partitioning $G$ as $G = [G_1 ~ G_2]$, where $G_1 \in \mathbb{R}^{T \times 
			n_x}$ and $G_2 \in \mathbb{R}^{T \times (n_z-n_x)}$, leads to
		\begin{align}\label{eq:lemma data closed pf2}
			x(k+1) &= (X_1 - D W_0)G Z(x(k)) + D w(k) \nonumber\\
			&= (X_1 - D W_0) [G_1 ~ G_2] \begin{bmatrix} x(k) \\ Q(x(k)) \end{bmatrix} + D 
			w(k) \nonumber\\
			&= \bar{A} x(k) + \bar{E} Q(x(k)) + D w(k),
		\end{align}
		where $\bar{A} = (X_1 - D W_0) G_1$ and $\bar{E} = (X_1 - D W_0) G_2$. 
	\end{proof}

Since the data-based closed-loop platooning error system \eqref{eq:lemma data closed 2} requires the unknown disturbance $w(k)$ and sequence $W_0$. Hence, a further discussion on the bounds of $W_0$ is recalled from \cite[Lemma 4]{Persis+23} and provided in Lemmas~\ref{lemma:disturbance bound}.  
	\begin{lemma}\label{lemma:disturbance bound}	
		Under Assumption \ref{assume:disturbance}, $W_0 \in \mathcal{W} := \{W \in \mathbb{R}^{n_w \times T} \mid W 
		W^\top \preceq \Delta \Delta^\top \}$, with $\Delta = \delta \sqrt{T} I_{n_w}$. 
		Given any 
		matrices $M \in \mathbb{R}^{n \times T}$ and $N \in \mathbb{R}^{n_w \times T}$ 
		and scalar $\epsilon > 0$, it holds that
		$	M W^\top N + N^\top W M \preceq \epsilon^{-1} M M^\top + \epsilon N^\top \Delta 
			\Delta^\top N, ~ \forall W \in \mathcal{W}. $
	\end{lemma}
	
	The proposed data-driven control is stated in
	Theorem~\ref{thm:control design}.
	\begin{theorem}\label{thm:control design}
		Under Assumptions \ref{assume:disturbance} and \ref{assume:nonlinearity},
		the platooning error system \eqref{eq:sys for design} is robustly stable by 
		applying the controller $u(k) = K Z(x(k))$ with
		\begin{equation}\label{controller:thm}
			K = U_0[Y ~ G_2] \begin{bmatrix}
				P & \mathbf{0}_{n_x \times (n_z-n_x)} \\ \star & I_{n_z-n_x}
			\end{bmatrix}^{-1},
		\end{equation} 
		if the following problem with the decision 
		variables $P \in \mathbb{R}^{n_x \times n_x}$, $Y \in \mathbb{R}^{T \times 
			n_x}$, $G_2 \in \mathbb{R}^{T \times (n_z - n_x)}$ and $\gamma$ is feasible:
		\begin{subequations}\label{op:control design}	
			\begin{align}
				&\hspace{1.6cm}  \underset{P, Y, G_2, \gamma}{\min}~\gamma \nonumber\\
				\label{const:control 1}
				&\text{subject~to:} ~ 
				P \succ 0, ~ \gamma > 0, \\
				\label{const:control 2}
				&\hspace{1.6cm} Z_0 Y = \begin{bmatrix} P \\ \mathbf{0}_{(n_z-n_x) \times n_x}\end{bmatrix}, \\
				\label{const:control 3}
				&\hspace{1.6cm} Z_0 G_2 = \begin{bmatrix} \mathbf{0}_{n_x \times (n_z-n_x)} \\ I_{n_z-n_x} \end{bmatrix}, \\
				\label{const:control 4}
				&\hspace{1.6cm} X_1 G_2 = \mathbf{0}, \\
				\label{const:control 5}
				&
				\setlength\arraycolsep{2pt}
				\begin{bmatrix}
					P  & \mathbf{0} & P & (X_1 Y)^\top & \mathbf{0} & Y^\top & \mathbf{0}\\
					\star & \gamma I_{n_w} & \mathbf{0} & \mathbf{0} & D^\top & \mathbf{0} & 
					\mathbf{0}   \\
					\star & \star & \gamma I_{n_x} & \mathbf{0} & \mathbf{0} & \mathbf{0} & 
					\mathbf{0} \\
					\star & \star & \star & \frac{\epsilon_1}{1 + \epsilon_1} P & \mathbf{0} & \mathbf{0} & D \Delta \\
					\star & \star & \star & \star & \frac{1}{\epsilon_1} P & \mathbf{0} & \mathbf{0}\\
					\star &  \star & \star & \star & \star & \epsilon_2 I_T & \mathbf{0} \\
					\star & \star & \star & \star & \star & \star & \frac{1}{\epsilon_2} I_{n_w}
				\end{bmatrix}
				\succ 0,
			\end{align}
		\end{subequations}	
		where $\epsilon_1, \epsilon_2 > 0$ are arbitrary scalars given by the user. 
	\end{theorem}
	\begin{proof}
		Suppose the SDP \eqref{op:control design} is feasible. Let $G_1 = Y P^{-1}$. The two constraints \eqref{const:control 2} and \eqref{const:control 3} together yield
		\begin{equation}\label{eq:thm control pf1}
			Z_0 [G_1 ~ G_2] = I_{n_z}.	
		\end{equation}
		Combining \eqref{eq:thm control pf1} with \eqref{controller:thm} gives
		\begin{equation}\label{eq:thm control pf2}
			\begin{bmatrix} K \\ I_{n_z}\end{bmatrix} = 
			\begin{bmatrix} U_0 \\ Z_0 \end{bmatrix} [G_1 ~ G_2].
		\end{equation}
		The satisfaction of \eqref{eq:thm control pf2} (\textit{i.e.}, \eqref{eq:lemma data closed 1}) allows the use of Lemma \ref{lemma:data-based closed-loop sys} and leads to the data-based closed-loop dynamics \eqref{eq:lemma data closed 2}. By further using the equality constraint
		\eqref{const:control 4}, \eqref{eq:lemma data closed 2} becomes
		\begin{equation}\label{eq:thm control pf4}
			x(k+1) = \bar{A} x(k) - D W_0 G_2 Q(x(k)) + D w(k).
		\end{equation}
		
		The next step is to prove that \eqref{const:control 5} ensures robust 
		asymptotic stability of \eqref{eq:thm control pf4} around the origin. Under 
		Assumption \ref{assume:nonlinearity}, the closed-loop dynamics are dominated by 
		the linear part. Hence, it is sufficient to analyse only the robust asymptotic stability 
		of the linear closed-loop dynamics
		\begin{equation}\label{eq:thm control pf5}
			x(k+1) = \bar{A} x(k) + D w(k).
		\end{equation}
		Consider the Lyapunov function $V(k) = x(k)^\top P^{-1} x(k)$. According to the Bounded Real Lemma \cite{SchererWeiland00}, \eqref{eq:thm control pf5} is $H_\infty$ robust asymptotically 
		stable if there exists a positive definite matrix 
		$P \in \mathbb{R}^{n_x \times n_x}$ and a scalar $\gamma > 0$ such that
		\begin{equation}\label{eq:thm control pf6}
			V(k+1) - V(k) + \gamma^{-1} x(k)^\top x(k) - \gamma w(k)^\top w(k) < 0.
		\end{equation}
		Applying \eqref{eq:thm control pf5} to \eqref{eq:thm control pf6} and rearranging the inequality gives
		\begin{align}\label{eq:thm control pf7}	
			\setlength{\parindent}{0in}	
			&x(k)^\top \left( \bar{A}^\top P^{-1}  \bar{A} - P^{-1} + \gamma^{-1} I_{n_x} \right) x(k) 
			\nonumber\\
			& + w(k)^\top (D^\top P^{-1} D - \gamma I_{n_w}) w(k) + x(k)^\top \bar{A}^\top 
			P^{-1} D w(k) \nonumber\\
			&  + w(k)^\top D^\top P^{-1} \bar{A} x(k) 
			< 0.
		\end{align}
		For any given scalar $\epsilon_1 > 0$, the following inequality holds:
		\begin{align}
			& x(k)^\top \bar{A}^\top P^{-1} D w(k) + w(k)^\top D^\top P^{-1} \bar{A} x(k) 
			\nonumber\\
			\leq{}& \epsilon_1^{-1} x(k)^\top \bar{A}^\top P^{-1}  \bar{A} x(k) + \epsilon_1 w(k)^\top D^\top P^{-1} D 
			w(k). \nonumber	
		\end{align}
		Then  
		a sufficient condition for \eqref{eq:thm control pf7} is given as
		\begin{align}\label{eq:thm control pf8}
			&x(k)^\top \left[ (1+\epsilon_1^{-1}) \bar{A}^\top P^{-1}  \bar{A} x(k) - P^{-1}  + \gamma^{-1} 
			I_{n_x} \right] x(k) \nonumber\\
			&+ w(k)^\top ( \epsilon_1 D^\top P^{-1} D - \gamma I_{n_w}) w(k) < 0.
		\end{align}
		
		Define $\xi(k) = [x(k)^\top, ~w(k)^\top]^\top$. The linear closed-loop 
		dynamics \eqref{eq:thm control pf4} are robustly stable if 
		\begin{equation}\label{eq:thm control pf9}
			\xi(k)^\top	\Pi \xi(k)
			< 0,
		\end{equation}	
		where $\Pi = \text{diag}(\Pi_{1,1},\Pi_{2,2})$, $\Pi_{1,1} = (1+\epsilon_1^{-1}) \bar{A}^\top P^{-1}  \bar{A} - P^{-1} + \gamma^{-1} I_{n_x}$ and $\Pi_{2,2} = \epsilon_1 D^\top P^{-1} D - \gamma I_{n_w}$.
		
		An equivalent condition to \eqref{eq:thm control pf9} is given as
$-\Pi \succ 0$. Applying Schur complement \cite{SchererWeiland00} to it yields
		\begin{equation}\label{eq:thm control pf11}	
			\begin{bmatrix}
				P^{-1}  & \mathbf{0} & I_{n_x} & \bar{A}^\top & 
				\mathbf{0}\\
				\star & \gamma I_{n_w} & \mathbf{0} & \mathbf{0} & D^\top  \\
				\star & \star & \gamma I_{n_x} & \mathbf{0} & \mathbf{0}\\
				\star & \star & \star & \frac{\epsilon_1}{1 + \epsilon_1} P & \mathbf{0} \\
				\star & \star & \star & \star & \frac{1}{\epsilon_1} P
			\end{bmatrix} 
			\succ 0.
		\end{equation}	
		Substituting $\bar{A} = (X_1 - D W_0) G_1$ into \eqref{eq:thm control pf11}, multiplying both its sides with $\text{diag}(P,I,I,I,I)$, using $G_1 = Y P^{-1}$, and then after some rearrangement, we can have that
%
		\begin{align}\label{eq:thm control pf13}
			& \Omega -  M W_0^\top N - N^\top W_0 M^\top \succ 0, \\
			\text{with}~\Omega &= 
			\begin{bmatrix}
				P  & \mathbf{0} & P & Y^\top X_1^\top  & \mathbf{0}\\
				\star & \gamma I_{n_w} & \mathbf{0} & \mathbf{0} & D^\top  \\
				\star & \star & \gamma I_{n_x} & \mathbf{0} & \mathbf{0}\\
				\star & \star & \star & \frac{\epsilon_1}{1 + \epsilon_1} P & \mathbf{0} \\
				\star & \star & \star & \star & \frac{1}{\epsilon_1} P
			\end{bmatrix}, \nonumber\\
			M^\top &= \left[ 
			Y, ~\mathbf{0}_{T \times n_w}, ~\mathbf{0}_{T \times n_x}, ~\mathbf{0}_{T \times n_x}, ~\mathbf{0}_{T \times n_x}
			\right], \nonumber\\
			N &= \left[ 
			\mathbf{0}_{n_w \times n_x}, ~\mathbf{0}_{n_w \times n_w}, ~\mathbf{0}_{n_w \times n_x}, ~D^\top, ~\mathbf{0}_{n_w \times n_x}
			\right]. \nonumber
		\end{align}
By using Lemma \ref{lemma:disturbance bound}, 
a sufficient condition to \eqref{eq:thm control pf13} is 
		\begin{equation}\label{eq:thm control pf14}
			\Omega - \epsilon_2^{-1} M M^\top - \epsilon_2 N^\top \Delta \Delta^\top N \succ 0,
		\end{equation}
for any given scalar $\epsilon_2 > 0$.		

Further applying Schur complement to \eqref{eq:thm control pf14} gives \eqref{const:control 5}. Therefore, the satisfaction of \eqref{const:control 5} leads to that of \eqref{eq:thm control pf6} and thus ensuring the robust asymptotic stability of \eqref{eq:thm control pf5}.
\end{proof}	
	
A condition ensuring feasibility of the SDP  
\eqref{op:control design} is that $Z_0$ has full row rank \cite{Persis+23}. 
This condition is necessary to have \eqref{const:control 2} and 
\eqref{const:control 3}, \textit{i.e.} \eqref{eq:thm control pf1}, fulfilled and it can be viewed as a condition on the richness of the data. 

The design in Theorem~\ref{thm:control design} ensures that the linear 
closed-loop dynamics \eqref{eq:thm control pf4} are robustly stable although 
the full closed-loop dynamics \eqref{eq:thm control pf4} has nonlinearity $D 
W_0 G_2 Q(x(k))$. 
The regions of attractions and robust invariant sets of the platooning error system \eqref{eq:sys for design} under the proposed controller can be characterised following the results in \cite[Secton VI]{Persis+23}.
It is also worth minimising the effect of nonlinearity 
during transients. This is achieved via modifying the objective 
function of the SDP \eqref{op:control design} to include the minimisation of both the values of $\gamma$ and $\| G_2 \|$. Hence, in practice the SDP \eqref{op:control design} is reformulated as
\begin{align}\label{op:control design2}	
	& \underset{P, Y, G_2, \gamma}{\min} ~\lambda_1 \gamma + \lambda_2 \|
	G_2 \| \nonumber\\
	\text{subject~to:} ~& 
	\eqref{const:control 1}, \eqref{const:control 2}, \eqref{const:control 3}, 
	\eqref{const:control 4}, \eqref{const:control 5},
\end{align}	
where $\lambda_1$ and $\lambda_2$ are given non-negative scalars. 

\begin{remark}\label{remark:runtime}
The SDP \eqref{op:control design2} is solved online only once for the entire platoon. To improve platooning performance, it is necessary to re-conduct the data collection and control design whenever a new platoon forms, \textit{e.g.}, due to cut-ins/outs of vehicles.	
Since the dimensions of decision variables increase with the number of vehicles, the SDP could be expensive to solve for large platoons. In such case, the onboard computational burden can be reduced by (virtually) splitting the platoon into small sub-platoons, for each an SDP problem of s smaller size can be formulated and solved. Moreover, it would be beneficial to solve the SDP using more powerful cloud computing facilities, if applicable.  

\end{remark}

	\section{CACC for Nonlinear Mixed Vehicle Platoon}\label{subsec:mixed platoon}
	This section applies the data-driven CACC to a mixed vehicle platoon with $n$ vehicles, $n_\text{av}$ AVs and $n_\text{hv}$ HVs. All vehicles are characterised by unknown nonlinear models and equipped with V2V communication.
	To ensure controllability of the mixed platoon, the HVs can be at any place in 
	the platoon except as the leader \cite{Lan+21d}. Let $\mathcal{N}_\text{av}$ and $\mathcal{N}_\text{hv}$ be the index sets of the AVs and HVs in the platoon, 
	respectively.  
	
	Each AV is modelled by the third-order nonlinear 
	system~\eqref{eq:nonlinear vehicle dynamics}. 
	The car-following behaviour of the $i$-th HV, $i \in \mathcal{N}_\text{hv}$, is 
	captured by the nonlinear system \cite{Lan+21d}:
	\begin{align} \label{HV dynamics}
		\begin{split}
			\dot{h}_i &= v_{i-1} - v_i, \\
			\dot{v}_i &= a_i, \\
			\dot{a}_i &= \left[ \alpha_i (V(h_i) - v_i) + \beta_i 
			(v_{i-1} - v_i) - a_i \right]/\tau_i,
		\end{split}
	\end{align}
	where $h_i = p_{i-1} - p_i$, 
	$\alpha_i$ is the headway gain, and $\beta_i$ is the relative velocity gain. 
	$V(h_i)$ is the spacing-dependent desired velocity given by 
	\begin{eqnarray} \label{range policy}
		V(h_i) = \left \{
		\begin{array}{cl}
			0, & h_i \leq h_\text{s} \\
			\frac{v_\text{max}}{2} [ 1 - \cos(\pi \frac{h_i - 
				h_\text{s}}{h_\text{g} - h_\text{s}}) ], & h_\text{s} < h_i < 
			h_\text{g} \\
			v_\text{max}, & h_i \geq h_\text{g} 
		\end{array}\right.,
	\end{eqnarray}
	where $h_\text{s}$ and $h_\text{g}$ are the gaps before the 
	HV intends to stop and to maintain the maximum velocity $v_\text{max}$, 
	respectively. 
	
The goal is to design $u_i$ for the $i$-th AV, $i \in \mathcal{N}_\text{av}$, ensuring the entire mixed vehicle platoon drive at a 
given desired constant velocity $v^* > 0$ whilst maintaining a desired inter-vehicular distance $h^* > 0$. The spacing errors $\tilde{h}_i$ and 
velocity errors $\tilde{v}_i$,  $i \in [1,n]$, are defined in the same way as 
in Section~\ref{sec:problem statement}. The platooning error system of the 
$i$-th AV, $i \in \mathcal{N}_\text{av}$, is described by \eqref{eq:error system 2}. 
The HV platooning error system of the $i$th HV, $i \in \mathcal{N}_\text{hv}$, 
is derived as
\begin{align}\label{mixeq:error system HV}
		\begin{bmatrix}
			\dot{\tilde{h}}_i \\ \dot{\tilde{v}}_i \\ \dot{a}_i
		\end{bmatrix}
		=&
		\underbrace{\begin{bmatrix}
				0 & -1 & 0 \\ 
				0 & 0 & 1 \\ 
				0 & -(\alpha_i + \beta_i)/\tau_i & -1/\tau_i
		\end{bmatrix}}_{A_i}	
		\begin{bmatrix}
			\tilde{h}_i \\ \tilde{v}_i 	\\ a_i
		\end{bmatrix} \nonumber\\
		& +
		\underbrace{\begin{bmatrix}
				0 & 1 & 0 \\0 & 0 & 0 \\0 & \beta_i/\tau_i & 0
		\end{bmatrix}}_{C_i}
		\begin{bmatrix}
			\tilde{h}_{i-1} \\ \tilde{v}_{i-1} \\ a_{i-1}
		\end{bmatrix} +
		\underbrace{\begin{bmatrix}
				0 \\ 0 \\ 1
		\end{bmatrix}}_{D_i} w_i,
\end{align} 
where $w_i =\alpha_i( V(h_i) - v^* )/\tau_i$. 
This HV platooning error system is derived without using the linearisation in the literature \cite{gao2016data,huang2019connected,huang2020learning,Lan+21d,Wang+22}, thus avoiding linearisation errors. 
	
Following Section~\ref{sec:problem statement}, a discrete-time platooning error system of the mixed platoon can be obtained in the form of \eqref{eq:sys for design}, but with $Q_j(x_j) = \mathbf{0}$, $u_j = 0$, $B_j = \mathbf{0}$, $E_j = \mathbf{0}$, for all $j \in \mathcal{N}_\text{hv}$.The same Assumptions \ref{assume:disturbance} and \ref{assume:nonlinearity} are made for the mixed vehicle platooning error system, but the disturbance bound $\delta$ is defined differently as follows: For $i \in \mathcal{N}_\text{av}$, $\delta_i$ are defined as in Assumption \ref{assume:disturbance}. For  $i \in \mathcal{N}_\text{hv}$, it is derived that $|w_i| \leq \alpha_i\max \{v_0, v_{\max} - v_0\}/\tau_i =: \delta_i, i \in \mathcal{N}_\text{hv}.$	Hence, choosing $\delta = \max_{i \in [1,n_w]} \delta_i$. The data-driven CACC design for the mixed vehicle platoon follows the same procedure as in Section~\ref{sec:control for AV platoon} and is thus not repeated here.
	
\section{Simulation Results} \label{sec:simulation}
This section reports comparative results of the proposed 
method and the classic ACC \cite{Shladover+12} in two simulation cases: 1) a platoon of pure 
AVs and 2) a mixed vehicle platoon. 
The classic ACC uses the control gains in the MATLAB example ``Adaptive Cruise Control with Sensor Fusion''.
Simulations are conducted in MATLAB running on Windows machine with a 12th Gen Intel(R) Core(TM) i7-1270P 2.2 GHz GPU and 16 GB RAM. The SDP problem is solved using the toolbox YALMIP 
\cite{lofberg2004yalmip} with the solver MOSEK \cite{mosek2010mosek}.
	
\textbf{Case 1: Platoon of pure AVs.}
This case studies a platoon with 4 AVs whose nominal vehicle parameters are 
\cite{Ghasemi+13}: $\tau_i = 0.2~\mathrm{s}$, 
$\sigma_i = 1~\mathrm{N/m^3}$, $M_i = 2.2~\mathrm{m^2}$, $c_i = 0.35$, $d_i = 
150~\mathrm{N}$, and $m_i = 1500~\mathrm{kg}$. To capture the vehicle 
heterogeneity and parameter uncertainties, a random deviation within $[-10\%, 10\%]$ is added to the nominal parameter values. 
The other platoon parameters are: $h^* = 20~\mathrm{m}$, 
$t_s = 0.05~\mathrm{s}$, $T = 500$.
The initial vehicle state $(p_i,v_i,a_i)$, $i \in [1,4]$, are randomly set as: 
$(65,20,0)$, $(40,15,0)$, $(25,18,0)$ and $(0,15,0)$, respectively. 
The desired velocity for the platoon to follow is shown in Fig. \ref{fig1}, which combines a 75~s constant speed driving at $v^*=20~\mathrm{m/s}$ (this period is set for forming the platoon and collecting data to design data-driven CACC) and the SFTP-US06 Drive Cycle. This velocity reference represents an aggressive, high speed and/or high acceleration driving behaviour with rapid speed fluctuations, which can validate the practical effectiveness of the proposed design.

\begin{figure}[t]
	\centering
	\includegraphics[scale=0.46]{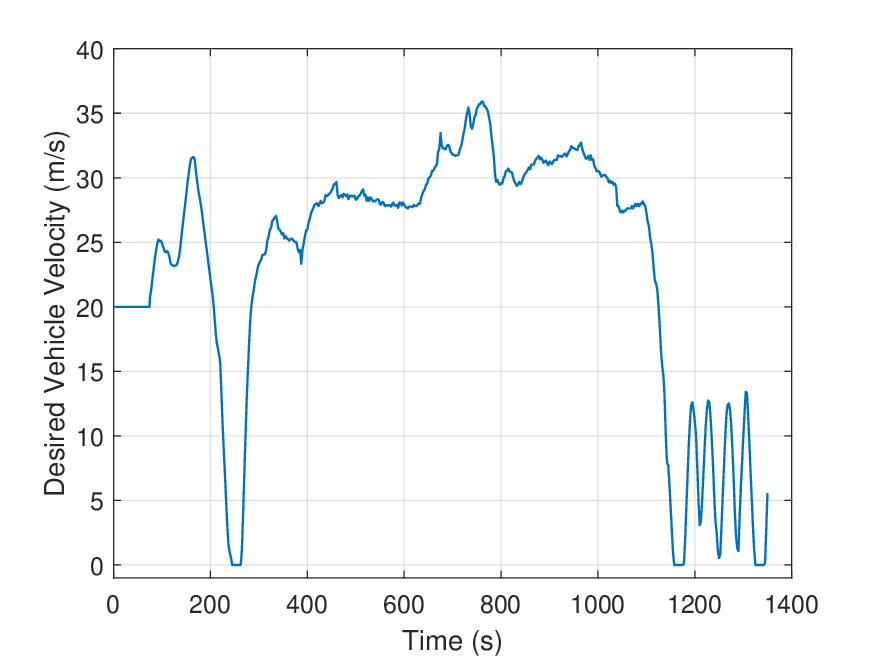} 
	\vspace{-2mm} 
	\caption{The desired velocity profile for the platoon.}
	\label{fig1}
\end{figure}

\begin{figure}[t]
	\centering
	\includegraphics[scale=0.5]{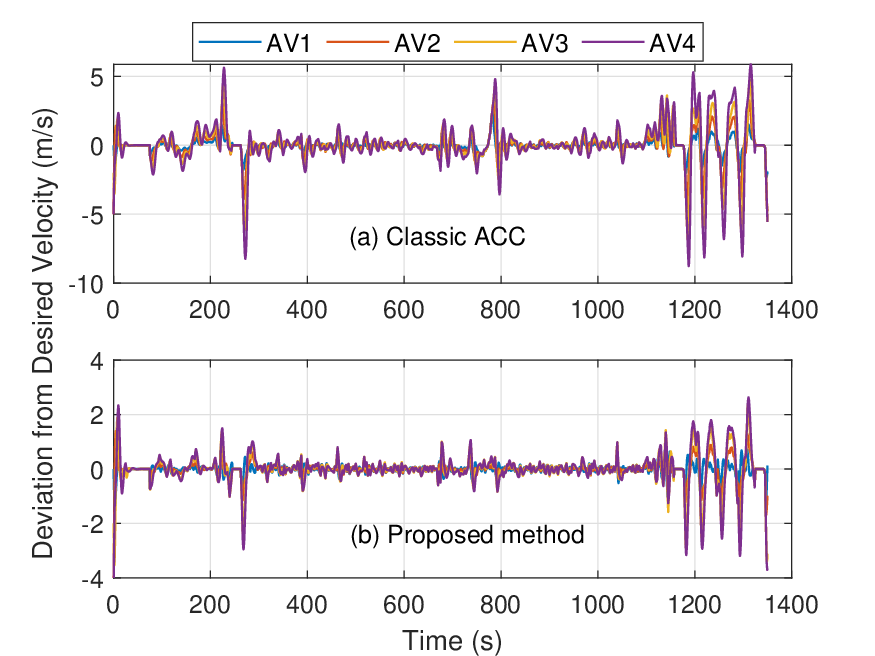}  
	\vspace{-2mm}
	\caption{Velocity deviations: Case 1.}
	\label{fig2}
\end{figure}

\begin{figure}[t]
	\centering
	\includegraphics[scale=0.5]{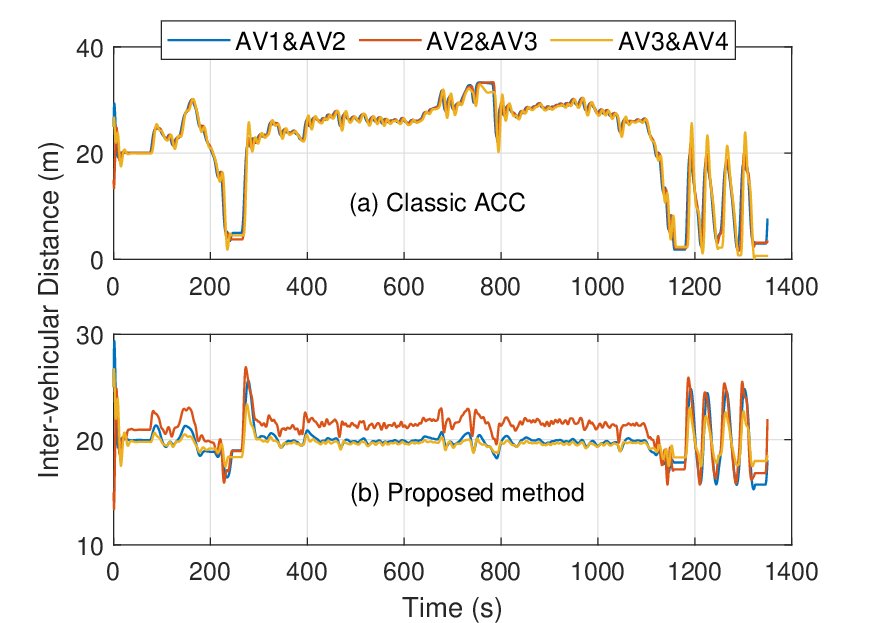} 
	\vspace{-2mm} 
	\caption{Inter-vehicular distances: Case 1.}
	\label{fig3}
\end{figure}
	
When designing controllers for all the 4 AVs using a single SDP problem \eqref{op:control design2}, solving the SDP needs 145.4 s. By splitting the platoon into two sub-platoons (one contains AV1\&AV2 and another contains AV3\&AV4), the control design is divided into two SDP problems each for a sub-platoon. Solving the two SDPs requires 17.1 s and 11.4 s, respectively. 
This confirms the discussions in Remark \ref{remark:runtime}.
	
The controllers solved from two sub-platoon SDP problems are implemented.
Fig. \ref{fig2} shows that the proposed method enables the platoon to follow 
the desired velocity profile closely and has smaller velocity deviations than 
the classic ACC. Consequently, the proposed method keeps the 
inter-vehicular distances close to the desired value $h^* = 20~\mathrm{m}$, as seen from Fig. \ref{fig3}.

\begin{figure}[t]
	\centering
	\includegraphics[scale=0.5]{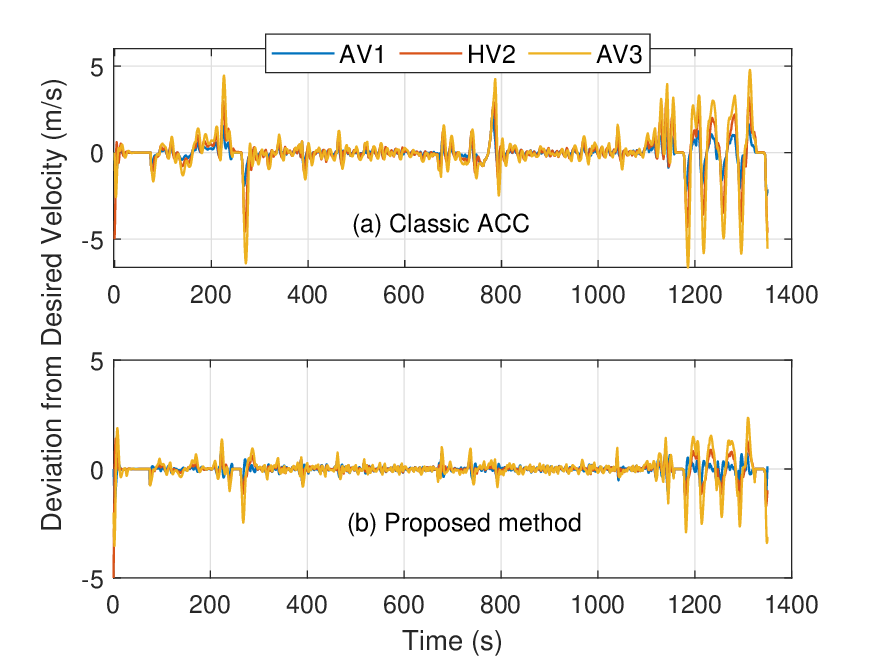}  
	\vspace{-2mm}
	\caption{Velocity deviations: Case 2.}
	\label{fig4}
\end{figure}

\begin{figure}[t]
	\centering
	\includegraphics[scale=0.5]{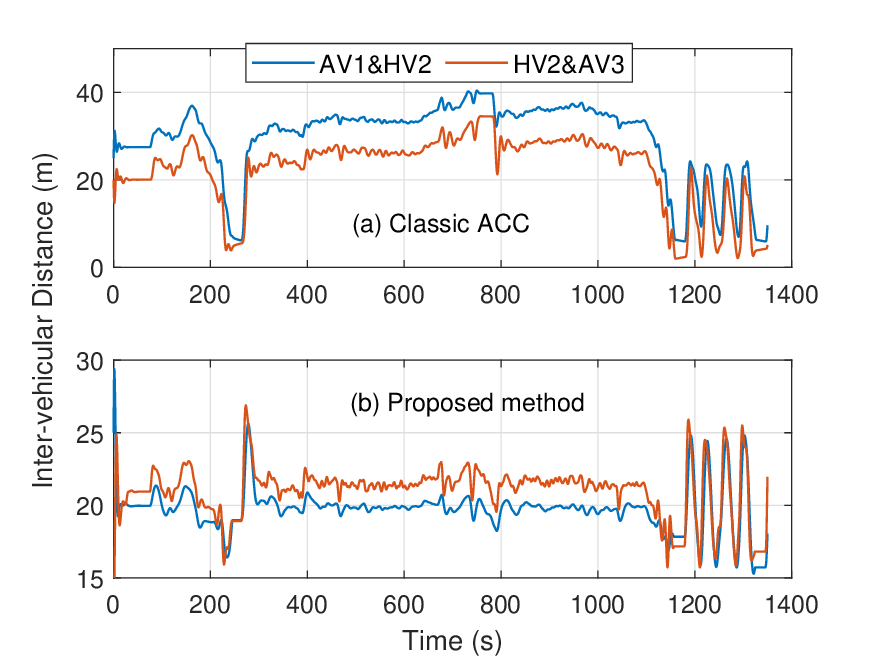} 
	\vspace{-2mm} 
	\caption{Inter-vehicular distances: Case 2.}
	\label{fig5}
\end{figure}
	
\textbf{Case 2: Mixed vehicle platoon.}
A mixed platoon with 3 vehicles, 2 AVs at the front and rear and an HV in the 
middle, is simulated. The AVs parameters are the same as Case 1. The HV parameters are $h_s = 5~\mathrm{m}$, $h_g = 50~\mathrm{m}$, $v_{\max} = 40~\mathrm{m/s}$, $\tau = 0.7~\mathrm{s}$, $\alpha = 0.2$, and $\beta = 0.4$. The initial vehicle state $(p_i,v_i,a_i)$, $i \in [1,3]$, are randomly set as: $(45,20,0)$, $(20,15,0)$ and $(0,20,0)$, respectively. The sampling time, number of data samples, desired inter-vehicular distance and velocity profile are the same as Case 1. 
	
Solving the SDP requires 18.2 s.  
As shown in Fig. \ref{fig4} and Fig. \ref{fig5}, compared to the classic ACC, the proposed method enables the platoon to follow the velocity profile more closely whilst having inter-vehicular distances that are closer to the desired value $h^* = 20~\mathrm{m}$.
	
\section{Conclusion} \label{sec:conclusion}
A data-driven CACC is proposed for vehicle platoons with consideration 
of the unknown nonlinear vehicle dynamics. The controller design is formulated 
	as an SDP problem that can be efficiently solved using 
	off-the-shelf solvers. The proposed data-driven control design is applicable to 
	platoons of pure AVs and also platoons of mixed AVs and HVs. The simulation 
	results demonstrate that the proposed method is more effective than the classic ACC in establishing a stable vehicle platoons under a representative aggressive velocity reference. Future work will focus on incorporating input and safety 
	constraints into the proposed data-driven control design to provide formal 
	guarantee on platoon safety.

	\bibliographystyle{IEEEtran} 
	\bibliography{IEEEabrv,References}

\end{document}